\documentclass[11pt]{article}%
\pdfoutput=1
\usepackage[hidelinks]{hyperref}
\usepackage{bbm,url,braket,microtype,amssymb,amsthm,mathtools,mathrsfs,cleveref,graphicx,float,authblk}
\usepackage{fullpage}
\usepackage{lmodern}
\usepackage[T1]{fontenc}
\usepackage[USenglish]{babel}
\usepackage[dvipsnames]{xcolor}
\usepackage[titletoc,title]{appendix}
\usepackage{subcaption}
\usepackage{multirow}

\providecommand{\U}[1]{\protect\rule{.1in}{.1in}}
\newtheorem{thm}{Theorem}\Crefname{thm}{Theorem}{Theorems}
\Crefname{lem}{Lemma}{Lemmas}
\newtheorem{prp}[thm]{Proposition}\Crefname{prp}{Proposition}{Propositions}
\Crefname{cor}{Corollary}{Corollaries}
\Crefname{prb}{Problem}{Problems}
\newtheorem{dfn}[thm]{Definition}\Crefname{dfn}{Definition}{Definitions}
\Crefname{rmk}{Remark}{Remarks}
\Crefname{section}{Section}{Sections}
\Crefname{appendix}{Appendix}{Appendices}
\numberwithin{equation}{section}
\let\oldref\ref
\renewcommand{\ref}[1]{(\oldref{#1})}
\DeclareMathOperator{\tr}{tr}
\DeclareMathOperator{\spec}{spec}
\newcommand{\haar}{\operatorname{Haar}}
\newcommand{\R}{\mathbb{R}}

\newcommand{\C}{\mathbb{C}}

\newcommand{\state}[1]{{\ket{#1}\bra{#1}}}

\newcommand{\norm}[1]{{\left\Vert{#1}\right\Vert}}

\newcommand{\E}{\mathbb{E}}

\newcommand{\cN}{\mathcal{N}}
\newcommand{\cC}{\mathcal{C}}
\newcommand{\cD}{\mathcal{D}}
\newcommand{\cL}{\mathcal{L}}

\title{Linear Cross Entropy Benchmarking with Clifford Circuits}

\author[1]{Jianxin Chen}
\author[2]{Dawei Ding}
\author[1]{Cupjin Huang}
\author[3]{Linghang Kong}
\affil[1]{Alibaba Quantum Laboratory, Alibaba Group USA, Bellevue, Washington 98004, USA}
\affil[2]{Alibaba Quantum Laboratory, Alibaba Group USA, Sunnyvale, California 94085, USA}
\affil[3]{Alibaba Quantum Laboratory, Alibaba Group, Hangzhou, Zhejiang 311121, P.R. China}

\date{}

\begin{document}

\maketitle
\begin{abstract}
With the advent of quantum processors exceeding $100$ qubits and the high engineering complexities involved, there is a need for holistically benchmarking the processor to have quality assurance. Linear cross-entropy benchmarking (XEB) has been used extensively for systems with $50$ or more qubits but is fundamentally limited in scale due to the exponentially large computational resources required for classical simulation. In this work we propose conducting linear XEB with Clifford circuits, a scheme we call \emph{Clifford XEB}. Since Clifford circuits can be simulated in polynomial time, Clifford XEB can be scaled to much larger systems. To validate this claim, we run numerical simulations for particular classes of Clifford circuits with noise and observe exponential decays. When noise levels are low, the decay rates are well-correlated with the noise of each cycle assuming a digital error model. We perform simulations of systems up to 1,225 qubits, where the classical processing task can be easily dealt with by a workstation. 
Furthermore, using the theoretical guarantees in~Chen {\em et al.} (arXiv:2203.12703), we prove that Clifford XEB with our proposed Clifford circuits must yield exponential decays under a general error model for sufficiently low errors. Our theoretical results explain some of the phenomena observed in the simulations and shed light on the behavior of general linear XEB experiments.
\end{abstract}

\tableofcontents

\section{Introduction}
Quantum computers are no longer theoretical constructs or even small-scale, proof-of-concept devices. There are already multiple processors with over 50 qubits~\cite{arute2019quantum,wu2021strong}, and recently a processor with more than 100 qubits has been built~\cite{collins_easterly_2021}. With this number of qubits comes a litany of engineering and operating challenges, including the overwhelming number of control lines and pulse generators required to perform gates and measurements. There are also challenges in processor design, including significant crosstalk between qubits and poor two-qubit gate connectivity. Given all these challenges, there is a demand for benchmarking these processors to have assurance that the algorithms we run will give accurate results. In addition, the figure of merit provided by such a benchmarking scheme can serve as an essential guide for hardware iteration to build ever higher quality processors. Although common methods such as single- and two-qubit randomized benchmarking (RB) can provide useful information, these methods are not sensitive to aggregate effects such as the difference between calibrating qubits on a processor individually or simultaneously~\cite{mckay2019three}. To capture such system-level effects, holistic benchmarks which should involve all or a large number of qubits on the processor are needed. Such benchmarks can also provide information about whether there are correlated errors in the system by comparing the results to the digital error model~\cite{arute2019quantum}. This is crucial for realizing fault-tolerant quantum computation.

To holistically benchmark modern quantum processors with tens of qubits, various schemes have been developed, most notably linear cross-entropy benchmarking (linear XEB). Linear XEB was originally proposed for the ``quantum supremacy'' experiment~\cite{arute2019quantum}, where it was used to characterize increasingly larger quantum circuits so as to extrapolate the error of the 20-cycle Sycamore circuit. To implement linear XEB, we run quantum circuits with random layers of gates and perform a computational basis measurement. We next classically compute the probability of the bit string we measure. We then repeat, taking an overall average. This average, up to constants, is the linear XEB measure. It has been experimentally and numerically observed that this measure exponentially decays with the number of cycles for a noisy circuit, and this decay exponent is proposed as a measure of gate quality~\cite{arute2019quantum,wu2021strong,gao2021limitations}. Although originally conceived to support the ``quantum supremacy'' claim, linear XEB has become a benchmarking scheme in its own right~\cite{arute2019quantum,wu2021strong,mi2021information,liu2021benchmarking}. Linear XEB has the advantage of requiring only a shallow circuit, which is easy to implement on current processors. However, if we want to benchmark larger devices, linear XEB has an obvious Achilles' heel. Namely, the random gate set is chosen such that the classical processing needed for linear XEB, that is, simulating the resulting random quantum circuits, scales poorly with the number of qubits. Indeed, by the quantum supremacy claim itself, linear XEB \emph{by design} cannot benchmark a processor with more than a few tens of qubits~\cite{arute2019quantum}.\footnote{Readers who are familiar with the ``quantum supremacy'' paper may point out that the paper's method of using patch circuits or elided circuits can push the qubit number further. However, using patch circuits is effectively benchmarking two separate processors. As for elided circuits, the complexity of classical simulation still scales poorly. } Therefore, if we want to keep apace with current hardware development, we need an alternative benchmarking scheme that can scale to larger qubit numbers. 

In this work, we propose a scalable benchmarking scheme which replaces the random circuits in  linear XEB with Clifford circuits. We call our scheme \emph{Clifford XEB}. This is a crucial modification since Clifford circuits are easy to classically simulate~\cite{gottesman1997stabilizer}. Moreover, Clifford XEB retains the convenience of requiring only shallow circuits. We believe that given the prevalence of Clifford RB~\cite{emerson2005scalable,levi2007efficient,helsen2020general} and linear XEB, it is a negligible overhead to adopt Clifford XEB as a new benchmarking scheme. To demonstrate the scalability of our scheme, we perform numerical simulations of certain random Clifford circuits on various topologies with depolarizing noise. We simulate circuits with up to 225 qubits. For noise levels comparable to what is possible in current hardware, we can see clear exponential decays whose rates are consistent with the digital error model. We also conduct simulations on the 2-D grid topology with 1,225 qubits to further showcase the scalability of Clifford XEB. On top of the promising numerical results, we also prove that Clifford XEB for the classes of circuits we consider yields an exponential decay under a general error model for sufficiently low error. This is necessary since we a priori do not know if we would even measure an exponential decay for every experimental setting we encounter (although this can be shown for special error models~\cite{liu2021benchmarking,gao2021limitations}). In addition to providing this assurance, our proof also explains some of the phenomena we observe in the simulations, such as the linear XEB measure rapidly decaying for low cycle numbers before entering a smooth exponential decay. Through our theoretical results we provide concrete guidelines for how experiments, both Clifford XEB \emph{and} conventional linear XEB, should be interpreted so as to correctly extract information about gate errors.

\section{The Clifford XEB Scheme}
\label{sec:numerics}
In this section we give an explicit description of the Clifford XEB scheme and perform numerical simulations showing its viability for benchmarking quantum processors with more than a thousand qubits.

\subsection{Description of the Scheme}

Consider a quantum computer with $n$ qubits. Let $S \subset SU(2^n)$ be a subset of the special unitary group and let $\mu$ be a probability distribution on $S$. The implementation map $\phi:S\to \cC(2^n)$, where $\cC(2^n)$ is the set of quantum channels on $n$ qubits, defines the noisy implementation of gates on the system. This in particular means that the noise for each Clifford element is independent of the previous elements, that is, the noise is Markovian. Assume the initial state is $\rho_0$ and final measurement POVM is $M=\{M_x\}_{x\in\{0,1\}^n}$, which are noisy realizations of some desired initial state and measurement. Then, a linear XEB experiment is defined as follows. 

\begin{dfn}[Linear XEB]
A linear XEB scheme with parameters $(S,\mu, \phi, M, \rho_0)$ is given by the following procedure:
\begin{enumerate}
    \item For a given number of cycles $m$, sample $m$ i.i.d. quantum gates $g_1,\ldots,g_m$ from $S$ according to the distribution $\mu$.
    \item Initialize the system to $\rho_0$ and apply $\phi(g_1),\phi(g_2),\ldots,\phi(g_m)$ to the system.
    \item Measure the system under a POVM $\{M_x\}_{x\in\{0,1\}^n}$ and get a binary string $x$. Calculate the probability of getting $x$ for the ideal circuit: $p_x(m)=|\braket{x|g_m\ldots g_2 g_1|0^n}|^2$.
    \item Repeat steps 2 to 3 multiple times for this sampled circuit and calculate the average of $p_x(m)$ over sampled strings $x$.
    \item Repeat steps 1 to 4 multiple times and let $\hat p(m)$ be the average of $\mathbb{E}_x[p_x(m)]$ over different circuits. Let $\hat q(m) = - 1+ 2^n \hat p(m) $.
    \item Repeat steps 1 to 5 for number of cycles $m_1,\ldots, m_k$. The returned results are $\hat q(m_1), \ldots, \hat q(m_k)$.
\end{enumerate}
\end{dfn}
\noindent We can see that $\hat q(m)$ is an unbiased estimator of the following quantity
\begin{equation}
    q_R(m) = -1+2^n \mathbb{E}_{g_1, \cdots, g_m \sim \mu} \Big[ \sum_{x \in \{0,1\}^m}|\braket{x|g_m\ldots g_2 g_1|0^n}|^2 \tilde p(x)\Big], \label{eq:xeb-def}
\end{equation}
where $\tilde p(x)$ is the probability of measuring the binary string $x$ in the noisy circuit:
\begin{equation}
    \tilde p(x) = \tr[M_x \cdot \phi(g_m) \circ \ldots \circ \phi(g_2) \circ \phi(g_1) (\rho_0)]. 
    \label{eq:noisy_prob}
\end{equation}
We next define Clifford XEB, a special case, illustrated in~\Cref{fig:circuit}.

\begin{dfn}[Clifford XEB]
A Clifford XEB scheme with parameters $(S,\mu,\phi,M,\rho_0)$ is a linear XEB scheme, where $S\subseteq \mathrm{Cl}(2^n)$ is a subset of the Clifford group.
\end{dfn}
\begin{figure}[h]
    \centering
    \includegraphics[width= 0.9\textwidth]{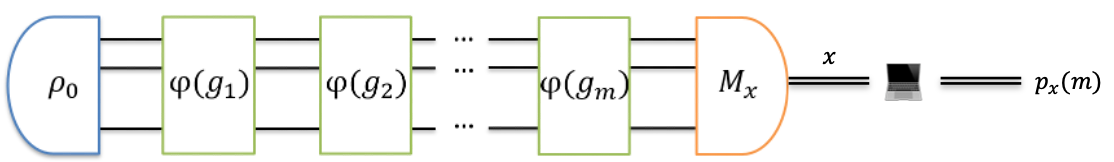}
    \caption{Illustration of one iteration of a Clifford XEB experiment. We start with an initial state, apply $m$ cycles of noisy implementations of Clifford elements, and finally perform a measurement. The measured bit string's ideal probability is then classically computed. This value is averaged over multiple measurements and random circuits. }
    \label{fig:circuit}
\end{figure}

That is, this Clifford XEB only differs from conventional linear XEB in that we require $S$ to be Clifford. Again, we do this since the ideal distributions of Clifford circuits can be efficiently computed classically~\cite{gottesman1997stabilizer}. Also note that our definition of cycles is slightly different from that of~\cite{arute2019quantum}. There they allow heterogeneous circuit geometries for every cycle, while we require that the gates in each cycle is sampled from the same distribution.

\subsection{Numerical Simulations}
\label{subsec:simulations}
We numerically study the behavior of Clifford XEB for various connectivity topologies to demonstrate its wide applicability. Our findings also help explain the behavior of general linear XEB schemes. The topologies we consider include a 1-D chain, a 2-D grid, and a star topology, i.e.\ there is a center qubit that is connected to all the other qubits, to implement the construction in~\cite{dankert2009exact}, which we will refer to as the \emph{Clifford approximate twirl}. We describe the particular distribution over Cliffords for each topology, illustrated in~\Cref{fig:topologies}.
\begin{figure}
\centering
    \begin{subfigure}[h]{0.31\textwidth}
        \includegraphics[width=\textwidth]{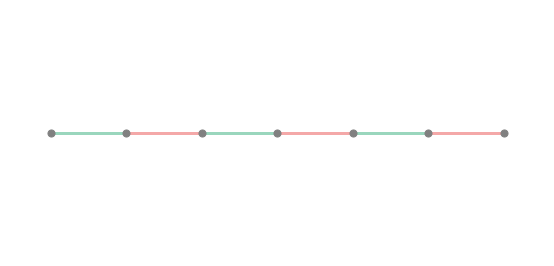}
        \caption{1-D chain.}
    \end{subfigure}
    ~
    \begin{subfigure}[h]{0.31\textwidth}
        \includegraphics[width=\textwidth]{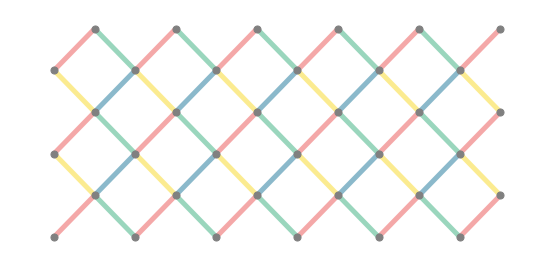}
        \caption{2-D grid.}
    \end{subfigure}
    ~
    \begin{subfigure}[h]{0.31\textwidth}
        \includegraphics[width=\textwidth]{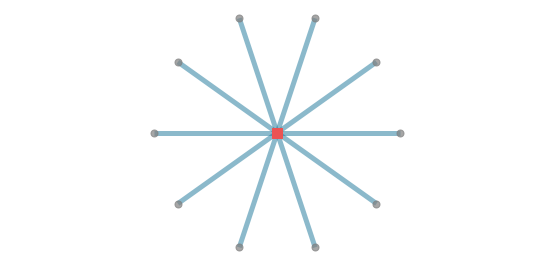}
        \caption{Star topology.}
    \end{subfigure}
    \caption{Illustration of different quantum device topologies we consider in our paper. In (a) and (b), qubit connections with different colors indicate different layers of parallel two-qubit gates to be applied. The Clifford approximate twirl can be natrually applied on a star topology depicted in (c).}
    \label{fig:topologies}
\end{figure}
\begin{enumerate}
    \item For the 1-D chain, each cycle consists of two layers of two-qubit gates each preceded by a layer of single qubit gates. The single-qubit gate layer involves i.i.d.\ random single-qubit Clifford gates on each qubit. The two-qubit layers are CNOT gates: the first connecting qubits $(0,1), (2,3), \cdots, (2k, 2k+1),\cdots$ and the second connecting $(1,2), (3,4),\cdots, (2k+1,2k+2), \cdots$.
    \item For the 2-D grid, we take motivation from the random circuits on the Sycamore processor in~\cite{arute2019quantum}. Each cycle consists of a single-qubit gate layer consisting of i.i.d.\ single-qubit Cliffords on every qubit, followed by a two-qubit gate layer consisting of CNOT gates. The structure of the CNOT's is chosen from the configurations A, B, C and D in~\cite[Figure 3]{arute2019quantum} with equal probability.\footnote{In the Sycamore random circuit sampling experiments, the two-qubit gates in different cycles follow a fixed pattern ABCDCDAB. We do not do this because we require every cycle to have the same distribution. }
    \item We implement the Clifford approximate twirl~\cite{dankert2009exact} on a star topology. Each cycle in this construction consists of two repeated circuits, each circuit involving three layers of two-qubit gates, interleaved by four layers of single-qubit gates. See~\Cref{app:dankert} for the detailed construction.
\end{enumerate}

The numerical experiments are run using the stabilizer-based Clifford circuit simulator Stim~\cite{gidney2021stim} on a 40-core Intel Xeon Platinum machine with 96GB memory. For each topology, we run experiments with different qubit numbers and noise levels. We pick qubit numbers 25, 100 and 225, which are respectively $5\times 5$, $10\times 10$, and $15\times 15$ on the 2-D grid. As for noise levels, we assume depolarizing noise for single- and two-qubit gates of $(10^{-5},10^{-4})$, $(10^{-4},10^{-3})$, and $(10^{-3},10^{-2})$ respectively. This is similar to what is achievable in current hardware. We also simulate the ideal case for reference. Each experiment is run for 2 to 50 cycles, where for each cycle number 3,000 random circuits are generated, each sampled 100,000 times. The results are plotted in~\Cref{fig:plot}. All of the experiments added together only take a total of several hours on the aforementioned workstation. In comparison, conventional linear XEB is already infeasible for just tens of qubits.

\begin{figure}
    \centering
    \begin{subfigure}[h]{\textwidth}
        \centering
        \includegraphics[width=\textwidth]{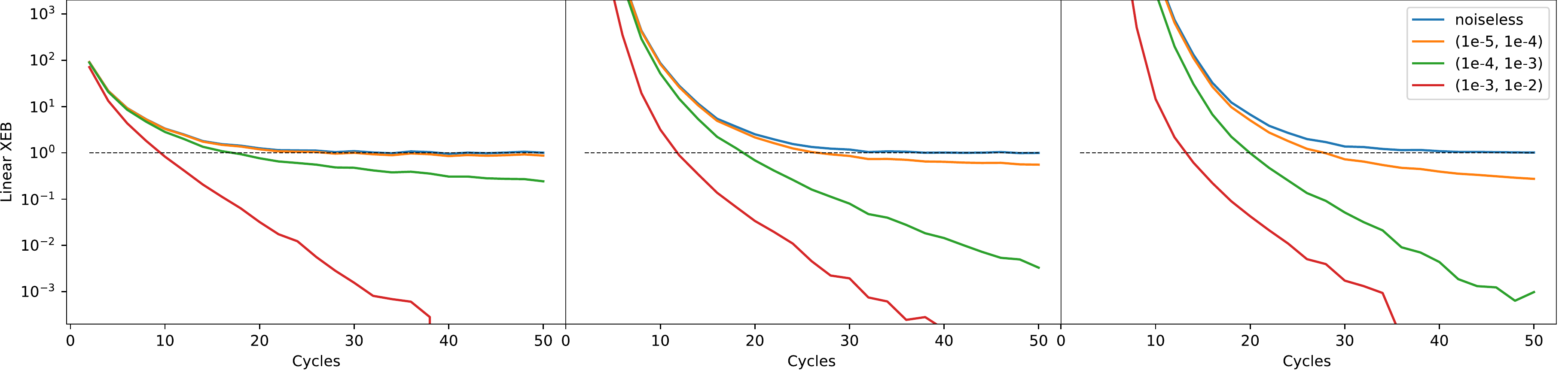}
        \caption{1-D chain}
        \label{fig:1d_chain}
    \end{subfigure}
    \begin{subfigure}[h]{\textwidth}
        \centering
        \includegraphics[width=\textwidth]{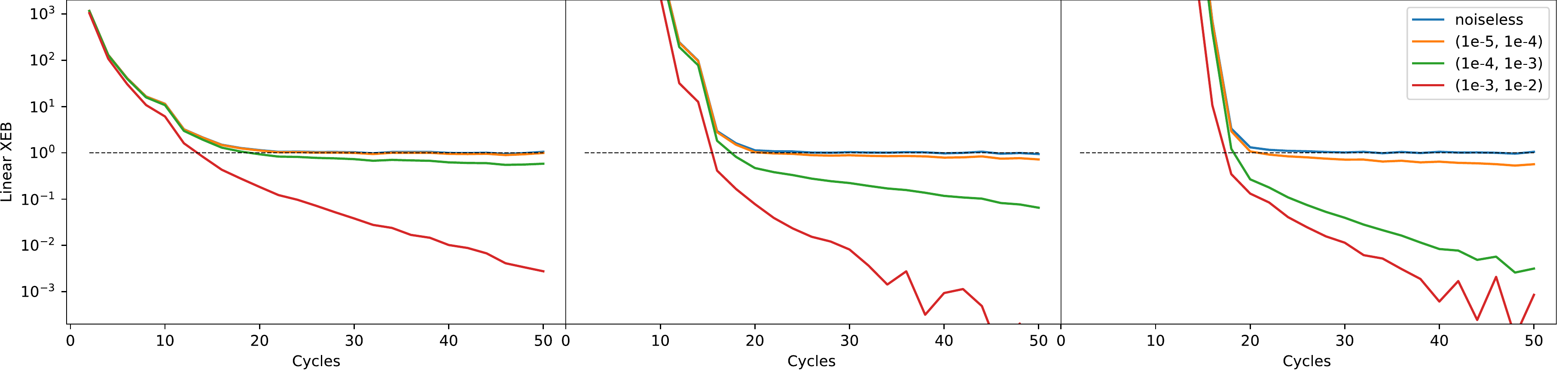}
        \caption{2-D grid}
        \label{fig:2d_grid}
    \end{subfigure}    
    \begin{subfigure}[h]{\textwidth}
        \centering
        \includegraphics[width=\textwidth]{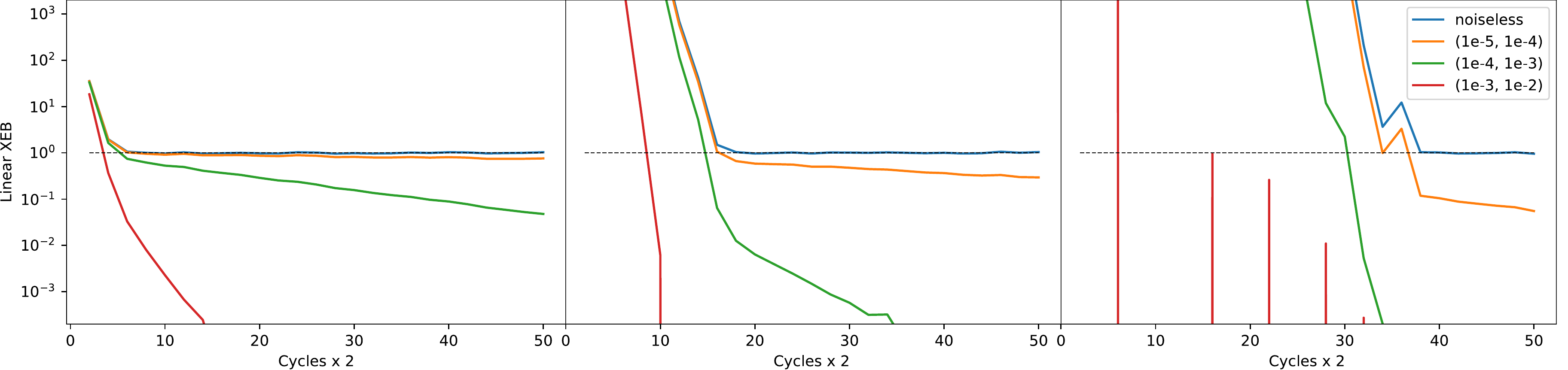}
        \caption{Clifford approximate twirl (star topology)\footnote{The $x$-axis is twice the cycle number since we need to apply the circuit construction twice to get a $\gamma$-approximate twirl. See~\Cref{sec:theory} for details. } }
        \label{fig:dankert}
    \end{subfigure}
    \caption{Linear XEB as a function of the number of cycles on different topologies, qubit numbers, and noise levels. Each row corresponds to the labeled topology. The number of qubits simulated on each topology is, from left to right, 25, 100, and 225. The horizontal line is linear XEB $=1$ for reference. }
    \label{fig:plot}
\end{figure}

\begin{figure}
    \centering
    \includegraphics[width=\textwidth]{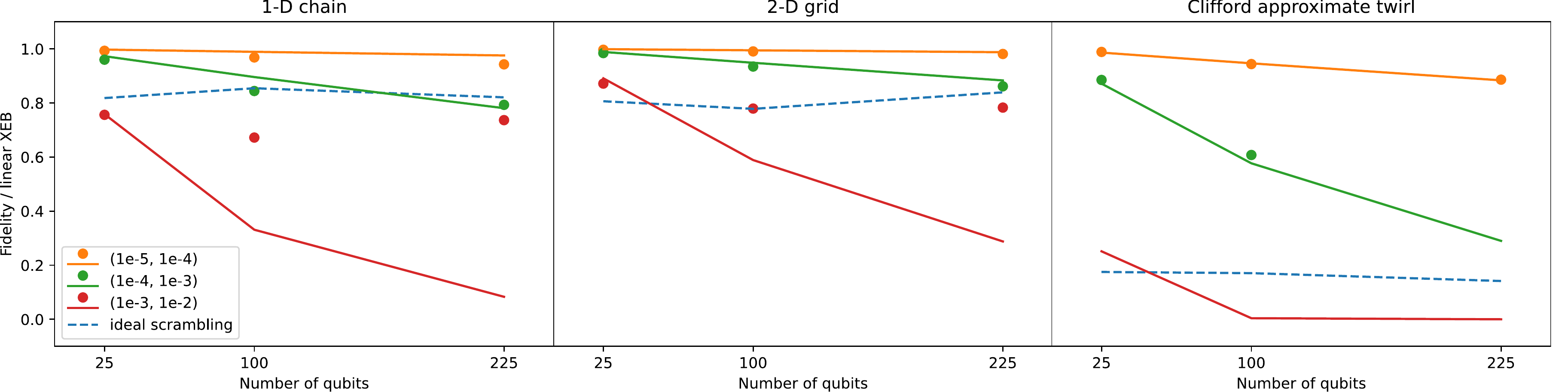}
    \caption{Comparison between the linear XEB decay rate, the digital error model, and the mixing rate, under different topologies. The blue dashed line depicts the decay rate of the ideal linear XEB towards 1. For each error rate, the solid lines are predictions given by the digital error model, and the dots are extracted from the numerical experiments. Some data is omitted due to failure of fitting.}
    \label{fig:error}
\end{figure}

From~\Cref{fig:plot} we can make the following observations.

\paragraph{Convergence and two-phase behavior}
We see a universal behavior of linear XEB with increasing cycle number: the ideal value always converges to 1,\footnote{Technically, the ideal linear XEB value converges to $\frac{D-1}{D+1}$ as shown in \Cref{app:var}, $D$ being the dimension of the quantum system. In the rest of the paper we take the approximation $\frac{D-1}{D+1}\approx 1$ unless explicitly stated.} while the noisy value converges to 0. This is the same as the experimental findings in~\cite{arute2019quantum}. Intuitively, for sufficiently many cycles, the random Clifford circuit will converge to a uniform distribution over the Clifford group~\cite{harrow2009random}, while a noisy implementation will eventually lose all information of the input. This explains the final convergence values. A more rigorous proof is given in~\Cref{app:var} and~\Cref{app:sed}.

The curves also exhibit a two-phase behavior: In the first phase, which we call the \emph{scrambling phase} in reference to the connection between random circuits and quantum scrambling~\cite{hayden2007black,sekino2008fast,lashkari2013towards,brown2012scrambling,shenker2014black, harrow2009random, harrow2018approximate, hunter2019unitary}, both the ideal and noisy linear XEB values decay quickly from being exponentially large to a constant value. The second phase, which we call the \emph{decaying phase}, the ideal linear XEB value converges to 1, while the noisy values undergo single exponential decays (which would be a linear decay on the semilog plot). Intuitively, only the data in the second phase is of relevance to benchmarking, as the first phase is characterized by the \emph{scrambling} of quantum information, rather than the \emph{loss} of it. We will give theoretical support for this intuition in~\Cref{sec:theory}.

One way of interpreting the decay rate is to use it as a proxy for the fidelity of a cycle, which in turn can be estimated using the digital error model \cite{arute2019quantum} by multiplying the fidelity of each gate in the circuit. To see if this applies to the numerical experiments, we extract the decay rates by linearly fitting the curve in the second phase (we manually choose the range of cycles). We then compare them with the digital error model in~\Cref{fig:error}. We see that the digital error model succeeds to predict the linear XEB decay rate when the error rate is small, but fails to do so when the error rate is too large. This is because even in the decaying phase, the ideal linear XEB term contributes an error term on top of the exponential decay from gate errors. As a result, when the gate error exceeds a certain threshold, the decay of the linear XEB value no longer represents the fidelity but is dominated by the slower mixing of the ideal term.

To verify this intuition, for each circuit configuration (topology and number of qubits) we sampled 300,000 random circuits to better evaluate the mixing rate of the ideal linear XEB value, which we call the \emph{mixing rate}. We do this because of the quick decay of the noiseless value and the high variance of Clifford XEB (See~\Cref{app:var}). The mixing rate towards $1$ is plotted as dashed lines in \Cref{fig:error}. It can be observed that the digital error model fails to predict the linear XEB decay rate when the fidelity estimated from the digital error model is below the line of mixing of the ideal linear XEB. This indicates that the Clifford XEB, and in general linear XEB, only provides useful information when the noise level is sufficiently small. As a rule of thumb, one can estimate the convergence rate of the ideal linear XEB values prior to running linear XEB experiments to ensure that the extracted decay rate from the experiment reflects the gate fidelities. To characterize large errors, one could either change the twirling distribution $\mu$ for faster scrambling (e.g. the one given in~\Cref{app:dankert} gives a theoretically provable mixing rate) or benchmark a subset of the qubits on the device to lower the error of each cycle.

\paragraph{Fluctuations}
Occasionally (e.g. in the rightmost plot of \Cref{fig:dankert}), there are noticeable bumps in the curves. Moreover we observe that the red line, corresponding to a two-qubit gate error rate of $10^{-2}$, fluctuates heavily when the linear XEB value is small. We attribute such instabilities to three main causes:
\begin{itemize}
    \item During the scrambling phase, both the expectation and the variance of the noiseless and noisy linear XEB value is large, resulting in significant fluctuations for the noiseless and noisy curves. The middle plot in~\Cref{fig:2d_grid} has such a bump at $m=14$.
    \item Unlike non-Clifford scrambling circuits that almost surely results in a measurement distribution obeying Porter-Thomas statistics, the measurement distribution of a Clifford circuit is always uniform on the support. As a result, the noiseless linear XEB value has constant variance even after the scrambling phase. This is different from doing linear XEB with Haar random gates because the Clifford group forms a unitary three-design but not a four-design. We derive the variance in \Cref{app:var}. Since the variance is constant, taking a large but constant number of random circuits smooths out the fluctuation. It can be observed in our plot that taking 3000 random circuits per data points effectively resolves this issue, except for a large deviation observed in the rightmost plot of \Cref{fig:dankert}, at 18 cycles.
    \item The linear XEB values of the noisy curves converge to 0 exponentially with respect to the cycle number. Given a fixed number of samples per random circuit, the absolute error of the linear XEB value is a constant, but the relative error quickly grows as the expectation approaches 0. Similar to~\cite{arute2019quantum}, sufficiently many random samples needs to be taken from each random circuit to ensure an accurate evaluation of the expectation value for linear XEB. This can be observed in the right two figures in \Cref{fig:2d_grid} where the red lines fluctuate at the tails, and in \Cref{fig:dankert} where the red lines are not even observed in our range. 
\end{itemize}

Understanding these sources of instability can help guide Clifford XEB and general linear XEB experiments. The second source only applies to Clifford XEB, but a constant number of random circuits suffices to reduce the variance to a small number. The third source applies for both Clifford XEB and linear XEB, in fact for any RB-based scheme. The number of samples taken must be sufficiently large such that the exponential decay curves can be confidently recovered. However, note that in our simulations we can often already isolate a clean single exponential decay at not too small linear XEB values for which we do not need as many samples.

To accurately extract decay rates, we need to perform Clifford XEB experiments in the decaying phase rather than the scrambling phase. A natural question to ask is at what cycle number a Clifford XEB scheme enters the decaying phase, i.e. the scrambling time. While we give evidence in~\Cref{sec:theory} and~\Cref{app:sed} that the scrambling time could scale linearly with respect to number of qubits, we observe experimentally that this only applies to Clifford approximate twirl (scrambling times are approximately 4, 9, 18 for 25, 100, 225 qubits, respectively). Scrambling happens much faster for the 1-D chain and 2-D grid experiments and seems only weakly dependent on the number of qubits. Although a finer theoretical analysis of the scrambling behavior of the noiseless experiment is needed to fully explain this phenomenon, we note that classical simulation of the noiseless experiment can be used to determine both the scrambling time and the mixing rate in the decaying phase. Such a simulation can be done prior to any physical experiment and serves as a guideline for the depths of the actual circuits to be run.

Finally, we present the results of Clifford XEB on the 2-D grid with $35\times 35=1225$ qubits, the results shown in~\Cref{fig:1225}. The experiments are run with two-qubit error rates ranging from $10^{-5}$ to $10^{-3}$ due to the larger number of qubits. We plot linear XEB for cycle numbers from $20$ to $100$ because of longer scrambling time. We see that, while the $10^{-3}$ error decay curve is dominated by the mixing rate, Clifford XEB is consistent with the digital error model when the error rate is below the mixing rate. This shows that the Clifford XEB scheme can be readily applied to quantum devices with more than 1,000 qubits. On the aforementioned workstation, all the noisy experiments took approximately 12 hours to run, and two days and a half to run the noiseless experiments in order to reliably extract the mixing rate.

\begin{figure}
    \centering
    \includegraphics[width=0.5\textwidth]{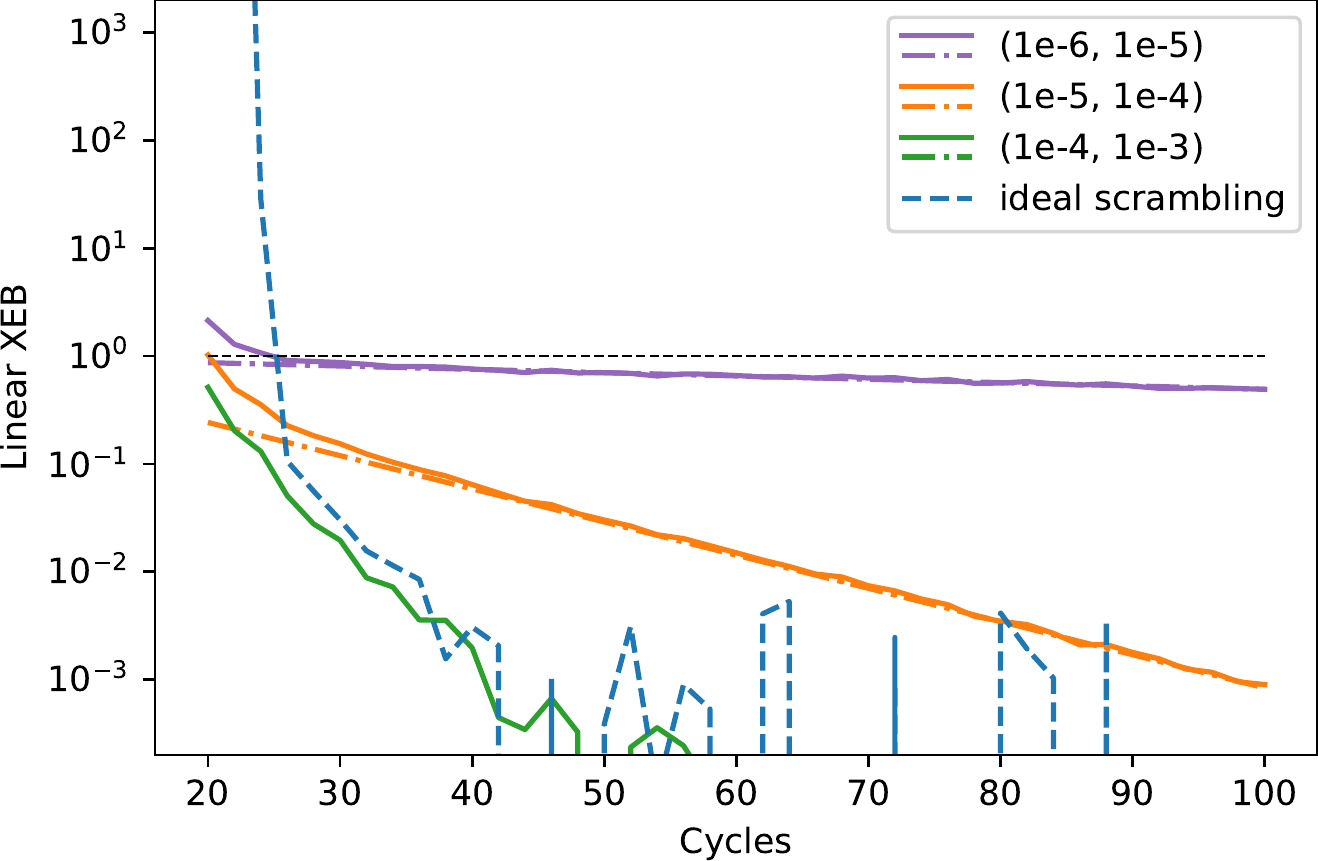}
    \caption{Clifford XEB on a 2-D grid of $35\times35=1225$ qubits. Data points for curves with noise are evaluated over 3000 random circuits with 100,000 samples each. Data points for the noiseless curve is evaluated over 300,000 random circuits. Predictions given by the digital error model is plotted in dashed dotted lines for comparison.}
    \label{fig:1225}
\end{figure}

\section{The Theory of Clifford XEB}
\label{sec:theory}
In this section we mathematically prove that Clifford XEB with the 1-D chain distribution, 2-D grid distribution, or the Clifford approximate twirl yields a single exponential decay under a general error model for sufficiently small errors. For the Clifford approximate twirl, this decay is theoretically guaranteed for a qubit number that scales inversely with the infidelity. Our results also support our explanations of phenomena observed in~\Cref{subsec:simulations}. These technical proofs can be skipped with little effect on the reading of the rest of the paper.

In~\cite{Chen2022}, a generalized framework for RB called universal randomized benchmarking (URB) is proposed, where they prove that for a particular class of schemes known as twirling schemes, an experiment would yield a single exponential decay. Here we show that Clifford XEB with our specific distributions are twirling schemes.

For any $g\in SU(2^n)$, define $\omega(g)$ as the corresponding quantum channel $\omega(g):\rho \to g \rho g^\dagger$. We first define a $\gamma$-approximate twirl, which is one possible definition of an approximate unitary 2-design~\cite{dankert2009exact,Chen2022}:
\begin{dfn}
\label{dfn:twirling}
Let $\mu$ be a measure on the unitary group and $\mathcal C$ be a linear operator on the space of Hermitian matrices.
The twirling map
\begin{equation}
    \Lambda(\mu): \cC \mapsto \int_{g\sim \mu} dg \, \omega(g^\dagger) \circ \cC \circ \omega(g)
\end{equation}
is a $\gamma$-approximate twirl if it satisfies
\begin{equation}
    |||\Lambda(\mu) - \Lambda(\mu_H)|||_\diamond \le \gamma,
\end{equation}
where $\mu_H$ is the Haar measure and $|||\cdot|||_\diamond$ is the induced diamond norm:
\begin{equation}
    |||\Lambda|||_\diamond := \max_{\cC} \frac{\norm{\Lambda(\cC)}_\diamond}{\norm{\cC}_\diamond}.
\end{equation}
Alternatively, $\Lambda(\mu)$ is a $\gamma$-approximate twirl in spectral norm if it satisfies 
\begin{equation}
    |||\Lambda(\mu) - \Lambda(\mu_H)|||_2 \le \gamma,
\end{equation}
where $|||\cdot|||_2$ is the spectral norm of the twirling map treated as a linear operator on superoperators.
\end{dfn}

\subsection{Clifford Approximate Twirl}
\subsubsection{Proof of Exponential Decay}
For Clifford approximate twirls, we use the following theorem in~\cite{Chen2022}:
\begin{thm}[Theorem~8 from~\cite{Chen2022} and \Cref{app:sed}]
\label{thm:urb}
Let $(S,\mu, \phi,M,\rho_0)$ be a linear XEB scheme on $n$ qubits. Suppose that the twirling map corresponding to $\mu$ is a $\gamma$-approximate twirl and that the implementation map satisfies
\begin{equation}
    \E_{g\sim \mu} \norm{\phi(g) - \omega(g)}_\diamond \le \delta.
    \label{eq:delta}
\end{equation}
If $\delta \le \frac{1 - \gamma}{11}$, then there exists $A,B\in \R$ and $p\in [1-2\delta,1]$, such that
\begin{equation}
    |q_R(m) - (A+Bp^m)| \le 16\times  2^n (\gamma+6\delta)^m. \label{eq:urb}
\end{equation}
Further, when $p<1$ we have $A=0$ and $B=O(1)$.
\end{thm}
\noindent Note that we give here a general result for linear XEB, that is, $S$ need not be Clifford. The main difference from the statement in~\cite{Chen2022} is the factor of $2^n$ in~\Cref{eq:urb}. This arises from the factor of $2^n$ in the definition of the XEB quantity in~\Cref{eq:xeb-def} which is used for normalization~\cite{arute2019quantum}. However, for large $n$ this is problematic since this would make the error term on the RHS in~\Cref{eq:urb} unacceptably large, considering that $B=O(1)$. To make the error term smaller than the exponential decay $A+Bp^m$ we want to extract, we need the number of cycles $m$ to scale linearly with $n$: $m=\Omega(n)$.\footnote{Note that~\Cref{eq:urb} holds for arbitrarily large $m$. Intuitively, the signal does not disappear at such large $m$ because we assume the noise between each cycle is Markovian, although error correlations are allowed to exist within a cycle.}  Indeed, we do observe that the scrambling time (end of the scrambling phase) is approximately a linear function of the number of qubits in~\Cref{fig:dankert}.

Now, we obtained the Clifford approximate twirl construction from~\cite{dankert2009exact}, where they proved that it constitutes a $\gamma$-approximate twirl. Hence, we can directly apply~\Cref{thm:urb} to Clifford XEB with the Clifford approximate twirl to theoretically guarantee an exponential decay. This guarantee holds under a general error model (we only assume Markovianity), giving much more confidence to what we can expect from running benchmarking experiments. A priori, for schemes where a particular error model is assumed or where there is no theoretical guarantee whatsoever, we might run an experiment and not even see an exponential decay, rendering the scheme completely inapplicable for benchmarking the device in question.

The Clifford approximate twirl is a theoretical construction with excellent circuit size and depth scaling, involving a probabilistic circuit of size $O(n\log 1/\gamma)$ and depth $O(\log n\log 1/\gamma)$. We prove in~\Cref{app:lb} that this is actually optimal in $n$ by showing any $\gamma$-approximate twirl requires a circuit of size $\Omega((1-\gamma) n)$ and depth $\Omega((1-\gamma) \log n)$. However, to achieve this scaling, we would need a complete graph topology. Fortunately, we can easily adapt it to other topologies that are more relevant to hardware at the cost of sub-optimal scaling. For a detailed analysis, see~\Cref{subsec:dankert_topologies}. 

\subsubsection{Qubit Scaling with Error}
Assuming a fixed gate error, any benchmarking scheme has a limit on the number of qubits it can be applied to. One reason for this simply follows from the size of the quantum benchmarking circuit required. To obtain a meaningful result from the circuit, the total error of the circuit should be less than unity. 

We consider the error of one cycle of a RB scheme. Now, we are working under a general error model where gate errors within a cycle can be arbitrarily correlated, so we need to consider worst-case additive\footnote{If errors are uncorrelated, we expect a multiplicative accumulation where fidelities of individual gates multiply, i.e.\ the digital error model. } error accumulation where errors of individual gates add. Then, we must require
\begin{equation}
\label{eq:deltaless1}
    s \epsilon \lesssim 1,
\end{equation}
where $s$ is the benchmarking circuit size and $\epsilon$ is the individual gate error (assumed to be the same for all gates). Since $s$ is a function of $n$, this gives an upper bound on $n$. In particular, for Clifford RB or any scheme where at least one uniformly random $n$-qubit Clifford element needs to be generated, $s = \Omega(n^2)$, and so
\begin{align}
    n \lesssim \frac{1}{\sqrt{\epsilon}}.
\end{align}
This is somewhat unsatisfactory; for instance, we would have to lower gate errors by a factor of 100 to be able to benchmark 10 times more qubits. 

We show in this section that Clifford XEB with the Clifford approximate twirl can do quadratically better: an inverse linear scaling $n \sim \epsilon^{-1}$. That is, by halving the gate error, we can benchmark twice as many qubits. This is actually optimal for twirling schemes which must implement $\gamma$-approximate twirls given the circuit lower bound we prove in~\Cref{app:lb}. In fact, as long as we have local gates (acts on a constant number of qubits), any benchmarking circuit in which all of the qubits are involved in a gate must have linear size. Thus, under a general error model, the best we can hope for in such a scenario is inverse linear scaling.

We now show achievability. The Clifford approximate twirl is able to form a $\gamma$-approximate twirl for $\gamma \approx 1/2$ with a linear size circuit~\cite{dankert2009exact}, so the bound on $n$ from that would be $n \lesssim \epsilon^{-1}$. However, we have an additional constraint. In~\Cref{thm:urb}, we require 
\begin{equation}
    \label{eq:delta_restriction}
    \delta \leq \frac{1-\gamma}{11}.
\end{equation}
This is a similar requirement to that of the main result in~\cite{helsen2020general}. Intuitively, $\gamma$ is the distance of our ideal twirling map from the Haar twirling map, which is known to give an exponential decay~\cite{dankert2009exact}. Hence, the higher $\gamma$ is, the lower the tolerance the exponential decay has to the error $\delta$, which is effectively a perturbation on the ideal twirling map. We analyze what~\Cref{eq:delta_restriction} implies about how the number of qubits we can benchmark scales with gate error for our scheme.

Recall $\delta$ is the error of an entire $\gamma$-approximate twirl, not just individual gate errors. We assume a simplified\footnote{This is for the sake of numerical convenience and \emph{not} a limiting assumption on what errors we allow. } error model where single-qubit gates are noiseless and two-qubit gates have the same error rate. Under additive error accumulation, $\delta$ is of the form 
\begin{align}
\label{eq:additive}
    \delta = s\epsilon,
\end{align}
where $\epsilon$ is the two-qubit gate error and $s$ is the size of the circuit. Now, for the Clifford approximate twirl, $s = c n k$, where $c$ is a constant that appears in the circuit construction, and $k = O(\log(1/\gamma)) \in \mathbb{N}$ is the number of times the circuit is repeated to achieve a $\gamma$-approximate twirl~\cite{dankert2009exact}. Explicitly, with $k$ repetitions, we achieve a $\gamma$-approximate twirl with
\begin{equation}
    \gamma = \frac{1}{2^{k-1}} \left(1+ \frac{1}{4^n}\right)^k.
\end{equation}
Hence,~\Cref{eq:delta_restriction} becomes
\begin{align}
\label{eq:before_approx}
    cn k\epsilon \leq \frac{1-\frac{1}{2^{k-1}} \left(1+ \frac{1}{4^n}\right)^k}{11}.
\end{align}

We want to look at the asymptotic case where $n$ is large, so we make a simplifying approximation $\frac{1}{4^n} \approx 0$. Then,~\Cref{eq:before_approx} becomes an upper bound on $n$:
\begin{equation}
    n \leq \frac{1-2^{-(k-1)}}{11ck\epsilon}.
\end{equation}
We see that for fixed $\epsilon$, to maximize the upper bound, we set $k=2$ or $k=3$ (they give the same answer) to obtain
\begin{align}
\label{eq:size_bound}
    n \leq \frac{1}{44c\epsilon}.
\end{align}
Thus, the qubit number scales inversely with the infidelity.

To be extremely concrete, we can fill in some numbers. State-of-the-art ion trap and superconducting circuit devices can achieve a two-qubit gate infidelity of around $10^{-3}$~\cite{krantz2019quantum,bruzewicz2019trapped}. Tracking the constants in~\cite{dankert2009exact}, $c=3$ for a star topology circuit. We therefore have a theoretical guarantee of an exponential decay for up to  $n = 7$ qubits\footnote{That is, $n=7, k=2, c= 3, \epsilon =  10^{-3}$ satisfies~\Cref{eq:before_approx}. } under a general error model. Note that there was little attempt in~\cite{Chen2022} to optimize the constant in~\Cref{eq:delta_restriction} or in~\cite{dankert2009exact} to optimize the circuit size constants, so the largest $n$ for which we have a theoretical guarantee may be substantially higher. We leave this for future work.

We compare our theoretical predictions with the findings in our numerical simulations. 
For the Clifford approximate twirl simulations, we only have a theoretical guarantee for the case with 25 qubits and two-qubit gate error rate $10^{-4}$ (yellow line in the leftmost plot in ~\Cref{fig:dankert}) if we fill in all the numbers and the constants. Nevertheless, most of the experiments give rise to single exponential decays that closely reflect the overall noise level of the circuits, which indicates that Clifford XEB is applicable to a wider variety of settings than might be suggested by theory. One reason for this disparity is that \Cref{thm:urb} applies to \emph{general} error models within each cycle, which can be highly correlated and, in the worst case, even adversarial. In contrast, in our numerical experiments we have gate-independent depolarizing errors. Since our scheme is optimal in scaling for arbitrary errors, it is unlikely that a theoretical guarantee for a general error model can be extended to over a hundred qubits given the two-qubit noise levels on current devices. To extend~\Cref{thm:urb} to larger quantum devices, we either have to reduce gate errors or make assumptions on the error models. For example, it is proven in~\cite{dalzell2021random} that the number of qubits can scale with $\epsilon^{-2}$, provided that the noise level is sufficiently weak and incoherent.\footnote{Note however they considered a setting different from our quest for exponential decays. } In general, we should seek a balance of theoretical guarantee and applicability, i.e., make a minimal number of realistic assumptions about errors to obtain a theoretical guarantee.

\subsection{1-D Chain and 2-D Grid}
We would hope that the theoretical analysis for the Clifford approximate twirl can extend to the schemes on the other topologies we consider. However, for the 1-D chain and 2-D grid schemes, the depth of each circuit cycle is constant. As shown in~\Cref{app:lb}, logarithmic depth is necessary to get an approximate twirl with $\gamma$ bounded away from $1$ by a constant. For these schemes, we instead argue that the measures on the unitary group they induce are $\gamma$-approximate twirls with respect to the $||| \cdot |||_2$ norm. In this case we can use the following theorem proved in~\cite{Chen2022}:
\begin{thm}
[Corollary 16 of~\cite{Chen2022}]
\label{thm:urb_spec}Let $(S,\mu, \phi,M,\rho_0)$ be a linear XEB scheme on $n$ qubits. Suppose that the twirling map corresponding to $\mu$ is a $\gamma$-approximate twirl with respect to the $|||\cdot |||_2$ norm and that the implementation map satisfies
\begin{equation}
    \E_{g\sim \mu} \norm{\phi(g) - \omega(g)}_\diamond \le \delta.
    \label{eq:delta_spec}
\end{equation}
Let $\delta':=2^{n/2}\delta$.
If $\delta' \le \frac{1 - \gamma}{11}$, then there exists $A,B\in \R$ and $p\in [1-2\delta',1]$, such that
\begin{equation}
    |q_R(m) - (A+Bp^m)| \le 16\times  2^{\frac 5 2 n} (\gamma+6\delta')^m. \label{eq:urb_spec}
\end{equation}
\end{thm}

\noindent Although the $2^n$ factors can be very large, this still shows that for sufficiently small errors, we can prove a single exponential decay under a general error model. 

In the case that the error channels are mixtures of unitaries, we have another version of \Cref{thm:urb_spec} which relaxes the noise level restriction to essentially that of \Cref{thm:urb}. We have the following.

\begin{thm}
\label{thm:urb_spec_mu}Let $(S,\mu, \phi,M,\rho_0)$ be a linear XEB scheme on $n$ qubits. Suppose that the twirling map corresponding to $\mu$ is a $\gamma$-approximate twirl with respect to the $|||\cdot |||_2$ norm, and that $\phi$ maps to probabilistic mixtures of unitaries:
$$\phi(g)=\int d\nu_g u$$
for probabilistic measures $\nu_g$ on $SU(d)$.
Additionally, suppose that the implementation map satisfies
\begin{equation}
    \E_{g\sim \mu} \E_{u\sim \nu_g} \norm{\phi(g) - u}_\diamond \le \delta.
    \label{eq:delta_spec_mu}
\end{equation}
If $\delta \le \frac{1 - \gamma}{11}$, then there exists $A,B\in \R$ and $p\in [1-2\delta,1]$, such that
\begin{equation}
    |q_R(m) - (A+Bp^m)| \le 16\times  2^{\frac 5 2 n} (\gamma+6\delta)^m. \label{eq:urb_spec_mu}
\end{equation}
\end{thm}

 Compared to \Cref{thm:urb_spec}, \Cref{thm:urb_spec_mu} eliminates the exponential prefactor $2^{n/2}$ on the noise level tolerance, with a slightly tighter restriction on the definition of the noise level $\delta$ itself. However, $\delta$ can be well approximated under an error model that can be readily expressed as mapping to mixtures of unitaries. We present the proof of  \Cref{thm:urb_spec_mu} in \Cref{sec:spec_norm}.

In general, it is much easier to prove a spectral gap than a gap with respect to the induced diamond norm. We have the following.

\begin{dfn}
\label{def:weakly_scrambling}
Let $\mu$ be a probabilistic measure on a finite group $G$, and $A:=\mathrm{Supp}(\mu)\subseteq G$ be its support. We say that $\mu$ is \emph{strongly scrambling} if $A$ is not contained in any (left or right) coset of a proper subgroup of $G$, or equivalently if $A^{-1}A$ generates $G$.
\end{dfn}
\begin{thm}

    Assume that $\mu$ is strongly scrambling over the $n$-qubit Clifford group. Then,
    \begin{equation}
        ||| \Lambda(\mu) - \Lambda(\mu_H) |||_2 < 1,
    \end{equation}
    where $\mu_H$ is Haar measure on $SU(2^n)$.
    \label{thm:gapped}
\end{thm}
\noindent The proof of~\Cref{thm:gapped} can be found in~\Cref{sec:spec_gapped}.  That the 1-D chain and 2-D grid Clifford XEB schemes satisfy~\Cref{def:weakly_scrambling} is straightforward and left as an exercise for the reader.

\section{Discussion}
\label{sec:discussion}
In this work we propose Clifford XEB, a benchmarking scheme for large-scale quantum devices using Clifford circuits. Clifford XEB is both efficient in quantum circuit size and necessary classical processing, and it has great promise for benchmarking quantum devices much larger than heretofore possible.  We numerically show that this scheme is feasible for more than a thousand qubits and theoretically prove that the scheme for the Clifford circuits we propose yields a single exponential decay for sufficiently small errors.

It is worthwhile to discuss in more detail why we need to holistically benchmark a quantum processor with tens or hundreds of qubits in the first place. One may think that it is sufficient to only perform single- or two-qubit benchmarking to characterize the native gate set. Even some effects of cross-talk can be captured via simultaneous single- or two-qubit RB~\cite{gambetta2012characterization}. The authors of~\cite{mckay2019three} study this topic in depth by performing Clifford RB on a 3-qubit processor. Interestingly they find that the 3-qubit RB result varies depending on the calibration procedure in a way that is not captured by simultaneous single- or two-qubit RB. This suggests there is additional information to be gained from running multiqubit benchmarks. This is especially relevant as the number of qubits on processors continues to increase and the calibration process becomes increasingly convoluted. Furthermore, there are a variety of obstacles to overcome to be able to scale quantum processors to achieve fault-tolerance, not the least of which the limited size of dilution fridges, and any workaround could introduce errors that can only be detected by a multiqubit benchmark. Lastly, using a multiqubit benchmark helps verify that the quantum processor does not have large correlated errors~\cite{arute2019quantum} which are detrimental to quantum error correction.

There have been other proposals for holistically benchmarking quantum computers. The quantum volume experiment~\cite{cross2019validating} compares different processors in terms of their overall computational capability rather than just gate quality, but it is not scalable with respect to the effective number of qubits. Many proposals for estimating the global Pauli errors~\cite{erhard2019characterizing, harper2021fast,flammia2021averaged,zhang2022scalable} provides more detailed characterizations of the error, but may be more suitable for detailed debugging rather than a holistic figure of merit. Additionally, there have been variants of RB for holistically benchmarking quantum computers in a scalable way, such as matchgate benchmarking~\cite{helsen2022matchgate}, direct randomized benchmarking~\cite{proctor2019direct}, and mirror randomized benchmarking~\cite{proctor2021scalable}. However, these schemes have not been as widely adopted in the experimental community as conventional linear XEB due to the latter's ease of hardware implementation. Our scheme boasts the exact same convenience but is also manifestly scalable. Another distinguishing feature is that we can prove Clifford XEB yields an exponential decay under a general error model. There is also constructions of exact unitary 2-designs using almost linear sized circuits (up to log factors)~\cite{cleve2016near}. Specifically, they require $O(n \log n \log \log n)$ size circuits. In contrast, the Clifford circuits we consider in our work are exactly linear size~\cite{dankert2009exact}. In fact, these log factors can still be quite large in practice (For example, if $n=50$, $\log_2 n \log_2 \log_2 n \approx 14$, a non-trivial overhead on the required error rates.). Finally, in~\cite{liu2021benchmarking}, there is a mention of conducting linear XEB with Clifford circuits. However, they do not pursue this line of thought nor further elaborate on its theory. Here we study Clifford XEB in depth from both a numerical and theoretical standpoint.

Similar to mirror RB~\cite{proctor2021scalable}, instead of classical simulation, we could sequentially reverse the Clifford layers on the quantum circuit and observe the exponential decay of the population of the initial state with respect to the number of cycles. The proof of exponential decay in~\cite{Chen2022} can be readily applied to this scheme as well. Compared to the Clifford XEB, this mirrored variant does not need to go through a scrambling phase and would probably need a smaller number of cycles. However, each cycle now consists of twice the number of gates compared to the Clifford XEB, and so the error tolerance of the protocol is halved compared to Clifford XEB. We leave a more detailed comparison for future work. Note that since it shares the same ideal twirling map as Clifford XEB, it could also potentially suffer from slow scrambling, making large gate error rates unextractable.

We give some directions for future research. Although numerical experiments in \Cref{sec:numerics} show that Clifford XEB can be used to benchmark large quantum processors up to a thousand qubits, ideally we want to implement Clifford XEB on an actual quantum device. Although there are currently few quantum processors with on the order of a hundred qubits, it would already be interesting to test Clifford XEB on real quantum devices with tens of qubits. Given the auspicious results from numerical simulations, we expect that Clifford XEB can benchmark quantum processors with more qubits than is promised by theory.

One important practical question is to build a benchmarking workflow based on Clifford XEB. Since it benchmarks circuits of linear size, the Clifford XEB can serve as an overall figure of merit of the entire processor but also serve as a guideline for subsequent, more detailed benchmarking schemes. Possible subsequent experiments include cycle benchmarking~\cite{erhard2019characterizing} or individual RB experiments. One can even apply Clifford XEB to subregions on a quantum chip as a measure of heterogeneous performance. We leave to future work to build such a practical workflow to efficiently obtain useful information about a quantum processor.

It would also be interesting to investigate, both from a theoretical and an experimental perspective, the mixing rate given a Clifford XEB scheme. Estimating the mixing rate is an essential step since it serves as a threshold below which the decay rate measured is not necessarily related to the gate error. For the numerical simulation on the $35\times 35$ grid, it took half a day to gather the noisy linear XEB results, but two days and a half to get the mixing rate. This is because the fluctuation of the ideal linear XEB value is purely from the randomness of the circuits chosen and is much larger than that of the noisy experiments. A faster algorithm that can extract the decay exponent without having to sample millions of random circuits would be needed to unblock this bottleneck. Towards this end, there have been studies on the spectral cap from random Clifford circuits from a Markov chain perspective~\cite{oliveira2007generic, vznidarivc2008exact, harrow2009random, brown2015decoupling,brandao2016local}. We leave it to future work to incorporate such methods to better estimate the spectral gap for particular random Clifford circuit ensembles. 

Of course, it would also be useful to improve the constants in the number of qubits we can benchmark given by~\Cref{eq:before_approx}. It would also be useful to realize a $\gamma$-approximate twirl in a way that is more natural for topologies that are more relevant for hardware, such as the 2-D grid. We could also consider topologies tailored for realizing quantum error correction, such as the honeycomb lattice~\cite{hastings2021dynamically}. Furthermore, we would love to see more work on how to make minimal assumptions on the gate errors to obtain a theoretical guarantee for a much higher number of qubits. For instance, most gates in superconducting circuits are limited by decoherence, which implies incoherent errors are dominant~\cite{krantz2019quantum}.

A possible alternative scheme to Clifford XEB is to use other gate sets that are classically simulable. If we can construct $\gamma$-approximate twirls using other gate sets, e.g., matchgates \cite{jozsa2008matchgates}, then we can formulate XEB schemes that might be more relevant in other experimental settings.

\section*{Acknowledgements} 
We would like to thank Yaoyun Shi for helpful discussions and comments. DD would like to thank God for all of His provisions.

\appendix

\section{Clifford Approximate Twirl Construction}
\label{app:dankert}
For completeness, we describe explicitly the sequence of gates to construct the Clifford approximate twirl given in \cite[Figure 1]{dankert2009exact}. 
\subsection{Description}
\label{subsec:dankert_description}
We first define the following procedures. Suppose we have $n$ qubits in a star topology, where we denote the center qubit as qubit 1. 
\begin{itemize}
    \item \textbf{Perform a $\mathcal P_n$ twirl}: Apply a random Pauli gate on each qubit.
    \item \textbf{Perform a $\mathcal C_1/\mathcal P_1$ twirl}: Apply a gate uniformly chosen from the set $\{I,S,H,SH,(SH)^2,SHS\}$, with $S$ being the single-qubit phase gate and $H$ the Hadamard gate. Alternatively this could be implemented by applying $(SH)^j$ with $j$ uniformly chosen from $\{0,1,2\}$, according to \cite{dankert2009exact}.
    \item \textbf{Conjugate qubit 1 by a random XOR}: For each qubit $2, \cdots, n$, apply a CNOT gate between this qubit and qubit 1 with probability 3/4.
\end{itemize}

\noindent We now describe how to implement the Clifford approximate twirl:

\begin{enumerate}
    \item Perform a $\mathcal P_n$ twirl.
    \item Perform $\mathcal C_1/\mathcal P_1$ twirl on all of the qubits.
    \item Conjugate qubit 1 by a random XOR.
    \item Apply $H$ to qubit 1 and $\mathcal C_1 / \mathcal P_1$ twirl the other qubits.
    \item Conjugate qubit 1 by a random XOR.
    \item Apply $H$ to qubit 1 and $\mathcal C_1 / \mathcal P_1$ twirl the other qubits.
    \item Apply $S$ to the first qubit with probability 1/2.
    \item Conjugate qubit 1 by a random XOR.
    \item $\mathcal C_1/\mathcal P_1$ twirl the first qubit.
    \item To obtain a $\gamma$-approximate twirl, repeat steps 2 to 9 for $O(\log(1/\gamma))$ times.
\end{enumerate}

\subsection{Different Topologies}
\label{subsec:dankert_topologies}
The only two-qubit gate operation involved in this construction is conjugating qubit 1 by a random XOR. This requires a star topology circuit. Furthermore, a direct implementation of this procedure requires a circuit depth of $\Omega(n)$. It is possible to implement it in depth $O(\log n)$ on a fully connected graph, as was shown in \cite[Figure 3]{dankert2009exact}, but realistic quantum processors have limited connectivity. To remedy this, we describe how to implement the random XOR circuit on a general connectivity graph. As a side result we also show that we can achieve depth $O(\log n)$ on a binary tree connectivity graph.

It is easy to see that in the computational basis, the conjugation by XOR operation effectively applies an X-gate to qubit 1 conditioned on the XOR of the random set. We therefore give a construction on a general graph by finding a spanning tree and iteratively computing the XOR of the random set through each level of the tree. Given the graph is connected, we can choose an arbitrary qubit to be the root (qubit 1), and form a spanning tree rooted from qubit 1. Each qubit $v$ apart from qubit 1 is then assigned a \emph{depth} $d(v)$ as its distance from qubit 1 on the spanning tree, and a \emph{degree} $\delta(v)$. Moreover, denote $D=\max_{v}d(v)$ and $\Delta=\max_v\delta(v)$ be the depth and degree of the spanning tree. Now, for each of the qubits other than qubit 1, add it to a set $S$ with probability 3/4. We call a qubit $v$ \emph{active}, if at least one of the qubits in the subtree rooted at $v$ (including $v$) is in $S$. We then perform the following operations:
\begin{enumerate}
    \item Do the following sequentially for each $l$ from $D-1$ to $1$:
    \begin{enumerate}
        \item Do the following in parallel for each active qubit $x$ at depth $l$:
        \begin{enumerate}
            \item If $x$ is not in $S$, it must have an active child $y$. Choose an arbitrary such $y$ and apply a CNOT gate from $x$ to $y$.
            \item Apply CNOT gates from each active child of $x$ to $x$ sequentially.
        \end{enumerate}
    \end{enumerate}
    \item Apply CNOT gates from all active qubits at depth 1 to the root sequentially.
    \item Apply the CNOT gates generated in step 1 again but in reverse order.
\end{enumerate}
To see why this circuit works, we first assume that every qubit is in the computational basis before the circuit, and the general case follows from linearity. For step $1$, one can see by induction that all qubits of depth $l'$ will have a value that equals to the XOR of qubits in its subtree that are contained in $S$ when the loop with $l=l'$ finishes. Step 1(a)i ensures that if $x$ is not in $S$, its value will be canceled after step 1(a)ii. Step 2 ensures that qubit 1 gets the corresponding XOR values from its children, and step 3 ``uncomputes'' the intermediate results in other qubits.

The total depth of the circuit can be upper bounded by $2-D\Delta$. To see this, observe that step 3 takes the same depth as step 1, which takes $(D-1)\cdot \Delta$ layers for sequentially executing $D-1$ repetitions of step 1(a), each takes at most $\Delta$ layers. Therefore, the depth of the circuit greatly depends on the spanning tree, which in turn depends on the underlying connectivity graph. When the spanning tree can be chosen to be a perfect binary tree, the random XOR can be done in $\Theta(\log(n))$ depth. For hardware relevant graphs such as the 1-d chain and the 2-d grid, the circuit depth is upper bounded by $O(D)$ as the degree of the spanning tree is upper bounded by the degree of the connectivity graph. In this case, the depth of the spanning tree can be minimized to the radius of the connectivity graph $r(G)=\min_u\max_v d(u,v)$ by rooting at $\arg\min_u\max_v d(u,v)$ and taking the spanning tree to consist of the shortest paths from each node to the root. The radius of the 1-d chain and the 2d grid is $\Theta(n)$ and $\Theta(\sqrt{n})$ respectively, giving a circuit depth for the XOR circuit $\Theta(n)$ and $\Theta(\sqrt{n})$ respectively. One can see that this is asymptotically tight as it takes $\Omega(r(G))$ steps to propagate information from an arbitrary qubit to a fixed one.

\section{Expectation and Variance of linear XEB}
It is observed in~\Cref{sec:numerics} that the Clifford XEB value converges to 1 and 0 for the ideal cases and the noisy cases respectively. Similar behavior has also been observed in linear XEB experiments~\cite{arute2019quantum, wu2021strong}. In this section we analyze the expectation and variance for ideal and noisy Clifford XEB values when the number of cycles is sufficiently large. We assume that for Clifford XEB and linear XEB, the first four moments of the distribution of the random gates converge to those of the uniform distribution over the Clifford group and the Haar measure over the unitary group respectively.

\subsection{Ideal Case}
\label{app:var}
For any circuit $C\in SU(2^n)$, we define
\begin{equation}
    \beta_C = 2^n \sum_x |\braket{x|C|0^n}|^4,
\end{equation}
which is the ideal linear XEB value plus 1. When $C$ is Haar random in $SU(2^n)$, $C\ket{0^n}$ is a uniformly random unit vector in $\C^{2^n}$. This vector can be parameterized as
\begin{equation}
    (e^{i\phi_1}\cos\theta_1, e^{\phi_2}\sin\theta_1\cos\theta_2, \ldots, e^{\phi_{D}}\sin\theta_1 \sin\theta_2 \ldots \sin \theta_{D-1})
\end{equation}
in the computational basis, where $D=2^n$ is the dimension. To compute the variance we claim we only need to consider the dependence on $\theta_1$ and $\theta_2$. The metric of the integral will contain two factors: the first factor is related to the spherical coordinates $\theta$ and is proportional to $\sin^{D-2}\theta_1 \sin^{D-3}\theta_2$, and the second factor is related to the phase factors and is equal to 
\begin{equation}
\cos\theta_1 \times \sin\theta_1\cos\theta_2 \times \sin\theta_1\sin\theta_2\cos\theta_3 \times \ldots  \propto
\cos\theta_1\cos\theta_2\sin^{D-1}\theta_1 \sin^{D-2}\theta_2.
\end{equation}
Let $t_1 =\cos^2\theta_1,t_2=\sin^2\theta_1 \cos^2\theta_2$, and the distribution of $t_1$ and $t_2$ is given by
\begin{align}
\Pr[t_1,t_2] \propto& \int d\theta_1 d\theta_2 \delta(t_1 -\cos^2\theta_1)\delta(t_2-\sin^2\theta_1 \cos^2\theta_2)\sin^{D-2}\theta_1 \sin^{D-3}\theta_2 \nonumber \\
&\times \cos\theta_1\cos\theta_2\sin^{D-1}\theta_1 \sin^{D-2}\theta_2 \nonumber \\
\propto& (1-t_1-t_2)^{D-3}.
\end{align}
for $0\le t_1,t_2 \le 1, t_1+t_2 \le 1$. One can normalize the distribution and have
\begin{equation}
    \Pr[t_1,t_2] = (D-1)(D-2)(1-t_1-t_2)^{D-3}.
\end{equation}
If we integrate over $t_2$, the distribution of $t_1$ is
\begin{equation}
    \Pr[t_1] = (D-1)(1-t_1)^{D-2}.
\end{equation}

By symmetry, for any $x\not= y$ the value of $|\bra{x}C\ket{0^n}|^2$ and $|\bra{y}C\ket{0^n}|^2$ will satisfy this distribution by identifying them with $t_1$ and $t_2$ respectively. Then one can calculate the mean and variance of $\beta_C$ as follows. 
\begin{align}
    \E_C \beta_C =& D \sum_x |\bra{x}C\ket{0^n}|^4 \nonumber\\
    =& D^2 \int_0^1 dt_1 \Pr[t_1]t_1^2 \nonumber \\
    =& \frac{2D}{D+1}, \label{eqn:expect}\\
    \E_C \beta_C^2 =& D^2\E_C \left(\sum_x |\bra{x}C\ket{0^n}|^4\right)^2 \nonumber \\
    =& D^2\E_C \sum_{x \not= y}|\bra{x}C\ket{0^n}|^4 |\bra{y}C\ket{0^n}|^4 + D^2\E_C \sum_{x}|\bra{x}C\ket{0^n}|^8 \nonumber \\
    =& D^3(D-1) \int_0^1 dt_1 \int_0^{1-t_1}dt_2 \Pr[t_1,t_2] t_1^2 t_2^2 + D^3 \int_0^1 dt_1 \Pr[t_1] t_1^4 \nonumber \\
    =& \frac{4D^2(D+5)}{(D+1)(D+2)(D+3)},
\end{align}
and so the variance is given by
\begin{equation}
    \E_C \beta_C^2 - \left(\E_C \beta_C\right)^2 = \frac{4D^2(D-1)}{(D+1)^2(D+2)(D+3)} = \frac{4}{D} + O(D^{-2}).
\end{equation}

From \Cref{eqn:expect} one can see that the ideal linear XEB value converges to $\frac{D-1}{D+1}$, which is close to 1 for large $D$.  This also applies to Clifford XEB since the uniform distribution on the Clifford group forms a unitary 2-design.

When $C$ is a uniformly random Clifford gate, $C\ket{0^n}$ will be a uniformly random stabilizer state. We introduce the concept of $k$-neighbor of a stabilizer state $\ket{\psi}$, which is the set of stabilizer states $\ket{\phi}$ such that $|\braket{\phi|\psi}|=2^{-k/2}$. It is known that every $n$-qubit stabilizer state has the same number of $k$-neighbors~\cite{garcia2014geometry}, which is denoted by $\cL_n(k)$. We know that for any stabilizer state $C\ket{0^n}$, the inner product $|\braket{x|C|0^n}|$ is $2^{-k/2}$ for $2^k$ choices of $x$, where $k$ is an integer, and for the remaining values of $x$ the inner product is 0. It is easy to see that $\beta_C = 2^{n-k}$ in this case.

For each $\ket{x}$, there are $\cL_n(k)$ $k$-neighbors, and each such state is in turn a $k$-neighbor of $2^k$ computational states. This means that there are $2^{n-k}\cL_n(k)$ stabilizer states that leads to $\beta_C = 2^{n-k}$. Let $\cN(n)$ be the total number of $n$-qubit stabilizer states. Then
\begin{align}
    \E_C \beta_C =& \sum_k 2^{n-k} \frac{2^{n-k}\cL_n(k)}{\cN(n)} = \sum_t 4^t \frac{\cL_n(n-t)}{\cN(n)}, \\
    \E_c \beta_C^2 =& \sum_k 4^{n-k} \frac{2^{n-k}\cL_n(k)}{\cN(n)} = \sum_t 8^t \frac{\cL_n(n-t)}{\cN(n)}.
\end{align}
It is known~\cite[Thm.~16]{garcia2014geometry} that in the large $n$ limit,
\begin{equation}
    c_1 2^{-t(t+5)/2} \le \frac{\cL_n(n-t)}{\cN(n)} \le c_2 2^{-t(t+3)/2}
\end{equation}
for some constants $c_1$ and $c_2$. So $\E_c \beta_C$ and $\E_c \beta_C^2$ are both $\Theta(1)$. This means that the variance of $\beta_C$ is $\Theta(1)$ when $C$ is a uniformly random Clifford circuit.

\subsection{Noisy Case}
\label{app:sed}
In this section, we investigate the single exponential decay behavior in the noisy case. In particular, we give estimates for the coefficients $A$ and $B$ in \Cref{thm:urb}, showing that $A=0$ and $B=O(1)$.

Following the universal randomized benchmarking (URB) framework~\cite{Chen2022}, we can write the Clifford XEB experiment results as
$$
    q_R(m)=-1+2^n \tr[|0^n\rangle\langle 0^n|(\Lambda_R^m(\mathcal{D}_M)(\rho_0))],$$
where the \emph{noisy twirling map} $\Lambda_R$ is defined as a map on quantum channels:
$$\Lambda_R: \mathcal{C}\mapsto \mathbb{E}_{g\sim u}[g^\dagger\circ \mathcal{C} \circ\phi(g)].$$

Now, using the notations in the proof of \cite[Theorem 8]{Chen2022}, $\Lambda_R$ is regarded as a perturbed version of $\Lambda^*_R$ and has a gapped spectrum leading to the single exponential decay in an appropriate parameter range. Specifically, define
$$X_1=\frac{\langle \cdot,\mathcal{D}\rangle_{SO}}{\langle \mathcal{D},\mathcal{D}\rangle_{SO}}\mathcal{D}+\frac{\langle\cdot, \mathcal{I}-\mathcal{D}\rangle_{SO} }{\langle\mathcal{I}-\mathcal{D}, \mathcal{I}-\mathcal{D}\rangle_{SO}}(\mathcal{I}-\mathcal{D}), X_2=\mathcal{I}-X_1,$$
One can verify that $X_i^2=X_i$ for $i=1,2$ and $$X_i \Lambda^*_R X_j=0$$ when $i\neq j$. The matrix perturbation result states that when $|||\Lambda^*_R-\Lambda_{\mathrm{Haar}}|||_\diamond <\gamma$, $|||\Lambda^*_R-\Lambda_R|||_\diamond\leq \delta$ and $\delta <\frac{1-\gamma}{11}$, there exists twirling maps $L,R$ such that \begin{itemize}
    \item $L=R^{-1}, ||| L|||_\diamond\leq 4, ||| R|||_\diamond\leq 4$;
    \item $X_i R\Lambda_R L X_j=0$ when $i\neq j$;
    \item $||| X_2R\Lambda_R LX_2|||_{\diamond}\leq \gamma+6\delta$;
    \item $||| X_1R\Lambda_R LX_1-X_1|||_{\diamond}\leq 2\delta$.
\end{itemize}

Let $A'_i=X_iR\Lambda_R LX_i$. It can be further proven that $A'_1$ has one eigenvalue $1$ and another real eigenvalue lying in $[1-2\delta, 1]$. Then $$\Lambda_R^m = L(A'_1)^mR+L(A'_2)^mR$$
and
\begin{align}
    q_R(m) &= -1+2^n \tr[|0^n\rangle\langle 0^n|(\Lambda_R^m(\mathcal{D}_M)(\rho_0))]\\
    &=-1+2^n \left[ \tr[|0^n\rangle\langle 0^n|(L(A'_1)^mR(\mathcal{D}_M)(\rho_0))]
    + \tr[|0^n\rangle\langle 0^n|(L(A'_2)^mR(\mathcal{D}_M)(\rho_0))]\right].
\end{align}
Since $$|||A'_2|||_\diamond\leq \gamma+6\delta,$$ the second term vanishes quickly, and the single exponential decay is given by the first term.

We assumed $p<1$, so $A$ is well-defined and is the contribution corresponding to the eigen-channel of $\Lambda_R$ with eigenvalue 1. It is easy to see that this eigen-channel is the depolarization channel $\mathcal{D}(\rho)=\frac{I}{D}$:
$$\Lambda_R(\mathcal{D})(\rho) = \int d\mu \omega^\dagger(g)\circ \mathcal{D}( \phi(g)(\rho))=\int d\mu \omega^\dagger(g)(\frac{I}{D})=\frac{I}{D}$$
since the inverting maps $\omega^\dagger(g)$ are all unital. We then have
$$A=2^n\tr[|0\rangle\langle 0|^{\otimes n} \mathcal{D}(\rho)]-1=0.$$
We thereby obtain the first conclusion $A=0$. Note that this holds in general when $\Lambda_R$ only has only one eigenvalue with modulus $1$: This ensures that $\Lambda_R^m(\mathcal{D}_M)\rightarrow \mathcal{D}$ when $m\rightarrow \infty$. This criterion does not pose any requirement on the state preparation or the measurement.

Both $\Lambda_R$ and $\Lambda^*_R$ are twirling maps; they map channels to channels. Consequently, they map differences of channels to differences of channels. It is thus more convenient to restrict them on the space spanned by the differences of channels. In this case one can define $X_1'$ to be
$$X_1' = \frac{\langle\cdot, \mathcal{I}-\mathcal{D}\rangle_{SO} }{\langle\mathcal{I}-\mathcal{D}, \mathcal{I}-\mathcal{D}\rangle_{SO}}(\mathcal{I}-\mathcal{D}).$$
The perturbation result tells us that 
\begin{align*}
    |B|&=|2^n\tr[|0\rangle\langle 0|^{\otimes n} L\cdot X_1'\cdot R(\mathcal{D}_M)(\rho)]|\\
&\leq 2^n \cdot \Vert L\cdot X_1'(R(\mathcal{D}_M))\Vert_{\diamond}\\
&\leq 2^n \Vert\frac{\langle R(\mathcal{D}_M), \mathcal{I}-\mathcal{D}\rangle_{SO}\cdot L(\mathcal{I}-\mathcal{D}) }{\langle\mathcal{I}-\mathcal{D}, \mathcal{I}-\mathcal{D}\rangle_{SO}}\Vert_\diamond\\
&= \frac{2^n}{\langle\mathcal{I}-\mathcal{D}, \mathcal{I}-\mathcal{D}\rangle_{SO}}\Vert L(\mathcal{I}-\mathcal{D})\Vert_\diamond |\langle \mathcal{D}_M, R^\dag(\mathcal{I}-\mathcal{D})\rangle_{SO}| .
\end{align*}
For sake of simplicity we denote $\mathcal{N}:=R^\dag(\mathcal{I}-\mathcal{D})$. We now bound the quantity $|\langle \mathcal{D}_M, \mathcal{N}\rangle_{SO}|$. 

Let $\mathcal{P}:\rho\rightarrow \sum_i \langle i|\rho|i\rangle |i\rangle\langle i|$ be the dephasing channel. One can check that it is self-adjoint with respect to the super operator inner product. Therefore
$$\langle \mathcal{D}_M, \mathcal{N}\rangle_{SO}=\langle \mathcal{P}\circ\mathcal{D}_M, \mathcal{N}\rangle_{SO}=\langle \mathcal{D}_M, \mathcal{P}\circ \mathcal{N}\rangle_{SO}.$$
As $\Vert \mathcal{P}\circ\mathcal{N}\Vert_\diamond\leq \Vert\mathcal{N}\Vert_\diamond$ due to data processing inequality, we can assume without generality that $\mathcal{N}$ is a quantum-classical superoperator: that is, there exists a set of operators $\{N_i\}_{i}$ such that
$$\mathcal{N}:\rho\rightarrow \sum_i\tr[N_i\rho]|i\rangle\langle i|.$$
Then 
\begin{align}
    \langle \mathcal{D}_M,\mathcal{N}\rangle_{SO}&=\sum_{P_i\in P^{\otimes n}}\tr[\mathcal{D}(P_i)\mathcal{N}(P_i)]\\
    &=\sum_{P_i\in P^{\otimes n}}\sum_{j\in\{0,1\}^n}\tr[M_jP_i][N_jP_i]\\
    &=2^n\sum_j\tr[M_jN_j]\\
    &=4^n \sum_j\frac{\tr[M_j]}{2^n}\cdot \tr[\bar{M_j}N_j],
\end{align}
where $\bar{M_i}=\frac{M_i}{\tr[M_i]}$. Since $\{M_i\}_i$ is a POVM, we have $M_i\succcurlyeq 0$ and $\sum_i M_i=I$, and therefore $\bar{M}_i$ is a quantum state. We argue that $|\langle \mathcal{D}_M,\mathcal{N}\rangle_{SO}|\leq 4^n \Vert\mathcal{N}\Vert_\diamond$. Assume otherwise, then there exists $i$ such that
$$|\tr[\bar{M_i} N_i]|>\Vert \mathcal{N}\Vert_\diamond.$$
However this is not possible as 
$$\Vert\mathcal{N}\Vert_\diamond \geq |\mathcal{N}(\bar{M_i})|_{\mathrm{tr}}=\sum_j |\tr[\bar{M_i}N_j]|\geq |\tr[\bar{M_i}N_i]|.$$
We then calculate $\langle \mathcal{I}-\mathcal{D},\mathcal{I}-\mathcal{D}\rangle_{SO}$. By linearity we have
\begin{align}
    &\langle \mathcal{I}-\mathcal{D},\mathcal{I}-\mathcal{D}\rangle_{SO}\\
    =&\langle \mathcal{I},\mathcal{I}\rangle_{SO}-2\langle \mathcal{D},\mathcal{I}\rangle_{SO} + \langle \mathcal{D},\mathcal{D}\rangle_{SO}.
\end{align}
For any quantum channel $\mathcal{C}$,
\begin{align}
    \langle \mathcal{D},\mathcal{C}\rangle_{SO}&=\sum_{P_i\in P^{\otimes n}} \tr[ \mathcal{D}(P_i)\mathcal{C}(P_i)]\\
    &=\sum_{P_i\in P^{\otimes n}} \tr[ 2^{-n}\cdot I\tr[P_i]\mathcal{C}(P_i)]\\
    &=\tr[2^{-n}\cdot 2^n \cdot \mathcal{C}(I)]=2^n.
\end{align}
On the other hand,
\begin{align}
    \langle \mathcal{I},\mathcal{I}\rangle_{SO}&=\sum_{P_i\in P^{\otimes n}} \tr[ \mathcal{I}(P_i)\mathcal{I}(P_i)]\\
    &=\sum_{P_i\in P^{\otimes n}} \tr[I]=4^n\cdot 2^n=8^n.
\end{align}
Therefore $$\langle \mathcal{I}-\mathcal{D},\mathcal{I}-\mathcal{D}\rangle_{SO}=8^n-2^n.$$

Plugging everything in we have
$$B\leq \Vert \mathcal{N}\Vert_\diamond\cdot \Vert L(\mathcal{I}-\mathcal{D})\Vert_\diamond\frac{4^n}{4^n-1}.$$
Since $\mathcal{N}=R^\dag(\mathcal{I}-\mathcal{D})$, we have 
$$|B|\leq |||L|||_\diamond \cdot|||R|||_\diamond\cdot \Vert \mathcal{I}-\mathcal{D}\Vert_\diamond^2\frac{4^n}{4^n-1}\leq 16\cdot 4\cdot 2\leq 128,$$
proving that $B=O(1)$.

\section{Circuit Lower Bounds for $\gamma$-Approximate Twirl}
\label{app:lb}
In this section we prove lower bounds on the circuit required to implement $\gamma$-approximate twirls. We state our result in the form of a theorem:
\begin{thm}
Suppose that the twirling map corresponding to a measure $\mu$ on $SU(2^n)$ is a $\gamma$-approximate twirl. Suppose we implement unitaries $g \in SU(2^n)$ sampled from $\mu$ by compiling them using single- and two-qubit gates, and let $s(g)$ and $t(g)$ be the number of gates and the depth for the circuit that implements $g$. Then
\begin{equation}
    \E_{g\sim \mu} s(g) = \Omega( (1-\gamma)n),\E_{g\sim \mu} t(g) = \Omega( (1-\gamma)\log n).
\end{equation}
\end{thm}
\begin{proof}
This proof is adapted from~\cite{cleve2016near}. We set the channel $\cN$ so that $\cN(\rho) = Z_1\rho Z_1$ where $Z_1$ is the Pauli $Z$ gate applied to the first qubit. Then we know that the twirling map $\Lambda_{\haar}$ maps $\cN$ to the following channel:
\begin{equation}
    \Lambda_{\haar} (\cN) : \rho \mapsto \frac{1}{4^n-1}\sum_{p \in \{0,1,2,3\}^n,p\not= 0^n} \sigma_p \rho \sigma_p = \frac{4^n}{4^n-1} \times \frac{I}{2^n}\tr[\rho] - \frac{1}{4^n-1} \rho.
\end{equation}
Let $\cD$ be the completely depolarizing channel on $n$ qubits
\begin{equation}
    \cD(\rho) = \frac{I}{2^n}\tr[\rho],
\end{equation}
then one can see that
\begin{equation}
    \norm{\cD - \Lambda_{\haar} (\cN) }_\diamond = O(4^{-n}).
\end{equation}

By definition
\begin{equation}
    \Lambda_\mu (\cN) : \rho \to \E_{g\sim \mu} g^\dagger Z_1 g \rho g^\dagger Z_1 g.
\end{equation}
Let $f(g)$ be the size of the support of $g^\dagger Z_1 g$, i.e., the number of qubits that $g^\dagger Z_1 g$ acts nontrivially on. For any $1\le j \le n$, define $p_j$ as the probability that qubit $j$ is in the support of $g^\dagger Z_1 g$ for $g$ sampled according to $\mu$. Assuming $p_j \le 1/2$, we have
\begin{align}
    \gamma \ge& \norm{\Lambda_\mu(\cN) - \Lambda_{\haar}(\cN)}_\diamond \nonumber\\
    \ge& \norm{\Lambda_\mu(\cN) - \cD}_\diamond - O(4^{-n}) \nonumber\\
    \ge& \norm{(\Lambda_\mu(\cN)-\cD)(\state{0^n})}_1 - O(4^{-n}) \nonumber\\
    \ge& \norm{\tr_{\bar j}[(\Lambda_\mu(\cN)-\cD)(\state{0^n})]}_1 - O(4^{-n}) \nonumber\\
    =& \norm{p_j \rho_j + (1-p_j)\state{0} - \frac{I}{2}}_1 - O(4^{-n}) \label{eq:step}\\
    \ge& \norm{p_j \state{1} + (1-p_j) \state{0} - \frac{I}{2}}_1- O(4^{-n}) \nonumber \\
    =& 1-2p_j- O(4^{-n}),
\end{align}
where $\tr_{\bar j}[\cdot]$ means tracing over all qubits other than qubit $j$, and in~\Cref{eq:step} $\rho_j$ refers to the state of qubit $j$ conditioned that some unitary is applied to $j$ in channel $\Lambda_\mu(\cN)$. This implies
\begin{equation}
    p_j \ge \frac{1-\gamma}{2}- O(4^{-n}). \label{eq:pj-bound}
\end{equation}
Obviously~\Cref{eq:pj-bound} is also satisfied when $p_j > 1/2$, so it always holds.
Then
\begin{equation}
    \E_{g\sim \mu} f(g) = \sum_j p_j \ge \frac{1-\gamma}{2}n - O(n4^{-n}).
\end{equation}
Note that for any $g$,
\begin{equation}
    s(g) \ge f(g)-1, t(g) \ge \log f(g) \ge \frac{f(g)-1}{n-1} \log n,
\end{equation}
so we have
\begin{equation}
    \E_{g\sim \mu} s(g) = \Omega( (1-\gamma)n),\E_{g\sim \mu} t(g) = \Omega( (1-\gamma)\log n).
\end{equation}
\end{proof}

\section{Spectral Norm Proofs}
\subsection{Proof of \Cref{thm:urb_spec_mu}}
\label{sec:spec_norm}
In this section we give a proof of~\Cref{thm:urb_spec_mu}. Specifically, following the proof of Theorem 16 in ~\cite{Chen2022}, It suffices to prove the following.

\begin{prp}
Let $\Lambda^*=\int d\mu \omega^\dagger(g)\circ\cdot \circ \omega(g)$ be the ideal twirling map associated to the probabilistic measure $\mu$, and $\Lambda=\int d\mu \omega^\dagger(g)\circ\cdot \circ \phi(g)$ be its noisy implementation. Moreover, suppose that $\phi$ maps to mixtures of unitaries, that is, for each $g$ there exists a probabilistic measure $\nu_g$ on $SU(d)$ such that
$$\phi(g)=\int d\nu_g u.$$
Then
$$|||\Lambda^*-\Lambda|||_2\leq \E_{g\sim \mu} \E_{u\sim \nu_g} \norm{\omega(g) - u}_\diamond.$$
\end{prp}

\begin{proof}
It can be verified that $$|||\Lambda^*-\Lambda|||_2\leq \int d\mu \norm{\omega(g)-\phi(g)}_2\leq \int d\mu \int d\nu_g \norm{\omega(g)-u}_2.$$
It then suffices to prove that $\norm{\omega(g)-u}_2\leq \norm{\omega(g)-u}_\mathrm{tr}$ for unitary channels $\omega(g)$ and $u$ since it is known that $\norm{\omega(g)-u}_\mathrm{tr}\leq \norm{\omega(g)-u}_\diamond$ .

Assume without loss of generality that $\omega(g)=\mathrm{Id}$ and $u=U\cdot U^{^\dag}$. Assume further that $U$ is diagonalized over the computational basis: $U=\sum_{j}e^{i\lambda_j} |j\rangle\langle j|$. Then the trace norm and the $2$-norm can then be computed explicitly as
$$\norm{\mathrm{Id}-u}_\mathrm{tr}=\norm{\mathrm{Id}-u}_\mathrm{2}=2\max_{i,j}\sin\frac{\lambda_i-\lambda_j}{2}.$$
\end{proof}

\subsection{Spectrally Gapped Schemes}
\label{sec:spec_gapped}

In this section we provide a proof to \Cref{thm:gapped}. Our main tool is random walks on finite groups~\cite{saloff2004random}.
We first make the following definition:
\begin{dfn}
Let $G$ be a finite group. We say a measure $\mu$ over $G$ is \emph{scrambling over $G$} if $A:=\mathrm{Supp}(\mu)$ generates $G$, and $A$ is not contained in a coset of a proper normal subgroup of $G$.
\end{dfn}
\noindent This is interesting because from~\cite{saloff2004random}, we can show that this measure converges to the uniform distribution under iteration, where iteration is defined as 
\begin{equation}
(\mu_1 * \mu_2)(g) := \int_{g_2 \sim \mu_2} dg_2 \mu_1(g g_2^{-1}) \mu_2(g_2).
\end{equation}
We state this fact as a proposition:
\begin{prp}[Proposition 2.3 and Theorem 2.1 of~\cite{saloff2004random}]
\label{prp:scrambling}
   Let $G$ be a finite group and let $\mu$ be a measure that is scrambling over $G$. Then,  $\mu$ converges to the uniform distribution $\mu_U$ over $G$ under iteration.
\end{prp}
\noindent We can now prove the following theorem.
\begin{thm}
    Let $\mu$ be an inverse symmetric measure (that is, $\mu(g) = \mu(g^{-1})$ for all $g$) that is scrambling over the $n$-qubit Clifford group $\mathcal C(n)$. Then,
    \begin{equation}
        ||| \Lambda(\mu) - \Lambda(\mu_H) |||_2 < 1,
    \end{equation}
    where $\mu_H$ is Haar measure on $SU(2^n)$.
    \label{thm:inverse_symm_gapped}
\end{thm}
\begin{proof}
Treating the twirling maps as linear maps on real superoperators, $\Lambda(\mu_H)$ is a projector onto the space spanned by the identity channel $\mathrm{id}$ and the completely depolarizing channel $\mathrm{dep}$. These two channels are also eigenvectors of $\Lambda(\mu)$ with eigenvalue 1. Since the singular values of $\Lambda(\mu)$ are at most 1, $|||\Lambda(\mu) - \Lambda(\mu_H)|||_2$ is therefore the third largest singular value of $\Lambda(\mu)$.

Now, by the argument in the proof of Theorem 8 in~\cite{Chen2022}, since $\mu$ is inverse symmetric, $\Lambda(\mu)$ is self-adjoint with respect to the SO inner product on real superoperators. Hence, it is diagonalizable by the spectral theorem. This implies the third largest singular value of $\Lambda(\mu)$ is third largest element of the elementwise absolute value of $\spec(\Lambda(\mu))$.

Next, it is easy to see that
\begin{equation}
    \Lambda(\mu^{* k}) = \Lambda(\mu)^k.
\end{equation}
Next, since $\Lambda(\mu)$ is diagonalizable, the eigenvalues of $\Lambda(\mu)^k$ are the $k$-th powers of the eigenvalues of $\Lambda(\mu)$. Two eigenvalues of $\Lambda(\mu)$ are 1, and every other eigenvalue has to be of magnitude at most 1. 

Assume for contradiction that there is a third eigenvalue of $\Lambda(\mu)$ with magnitude 1. Then, for any $k$, $\Lambda(\mu)^k$ always has at least three eigenvalues with magnitude 1. However, by~\Cref{prp:scrambling}, for all $g \in \mathcal C(n)$,
\begin{equation}
    \lim_{k\to \infty} \mu^{* k}(g) - \mu_U(g) = 0.
\end{equation}
Thus, for arbitrarily large $k$, the difference is arbitrarily small. Furthermore,
\begin{equation}
    ||| \Lambda(\mu^{* k}) - \Lambda(\mu_U) |||_2 \leq \max_{g \in \mathcal C(n)}  (\mu^{* k}(g) - \mu_U(g)) \vert \mathcal C(n) \vert.
\end{equation}
Thus, for arbitrarily large $k$, the $|||\cdot|||_2$ norm is also arbitrarily small. Now, since groups are closed under inverses, $\mu_U$ is symmetric under inverse. Thus, $\Lambda(\mu_U)$ is self-adjoint. Furthermore, since $\Lambda(\mu)$ is self-adjoint, $\Lambda(\mu^{* k})= \Lambda(\mu)^k$ is also self-adjoint. By Weyl's inequality on matrix perturbation, the eigenvalues of $\Lambda(\mu^{* k})$ and $\Lambda(\mu_U)$ can be paired up with a difference of at most $|||\Lambda(\mu^{* k})- \Lambda(\mu_U)|||_2$. However, this can be arbitrarily small as $k \to \infty$, and since $\mathcal C(n)$ is a unitary 2-design, two eigenvalues of  $\Lambda(\mu_U)=\Lambda(\mu_H)$ are 1 and the rest are 0. This clearly contradicts $\Lambda(\mu)^k$ having at least three eigenvalues with magnitude 1. Hence we conclude $\Lambda(\mu)$ has only two eigenvalues with magnitude 1. 

This means the third largest element of the elementwise absolute value of $\spec(\Lambda(\mu))$ is strictly less than 1, which is the desired conclusion.
\end{proof}

\noindent We can actually leverage this theorem to strongly scrambling measures $\mu$ that are not inverse symmetric. Define the inverse measure
\begin{equation}
    \mu^{-1}(g) := \mu(g^{-1}).
\end{equation}
We first show the following.
\begin{prp}
Let $\mu$ be a probabilistic measure on a finite group $G$. Let $A:=\mathrm{Supp(\mu)}$. Then the following are equivalent:
\begin{enumerate}
    \item $A$ does not lie in any (left or right) coset of a proper subgroup of $G$;
    \item $\langle A^{-1}A\rangle=G$;
    \item $\mu^{-1}*\mu$ is scrambling over $G$.
\end{enumerate}
\label{prp:weakly_scrambling}
\end{prp}

\begin{proof}
We follow the directions $1\Leftrightarrow 2\Leftrightarrow 3$.

\begin{itemize}
    \item $1\Rightarrow 2$. Assume otherwise that $\langle A^{-1}A\rangle=H<G$. Then for any $a,b\in A$, we have $$a^{-1}b\in H\Rightarrow b\in aH\Rightarrow A\subseteq aH,$$
contradicting with $1$.
\item $2\Rightarrow 1$. Note that $ aH=(aHa^{-1})a$, and a left cost of a subgroup is at the same time a right coset of a conjugate subgroup. Assume that $A\subseteq aH$ for some $a\in G, H<G$. Then $\langle A^{-1}A\rangle \subseteq H<G$, contradicting with $2$. 
\item $2\Rightarrow 3$. It is easy to verify that $\mathrm{Supp}(\mu^{-1}*\mu)=A^{-1}A$. It suffices to prove that $A^{-1}A$ does not lie in any coset $aH$ of a proper normal subgroup $H \trianglelefteq G$. Assume otherwise, then either $1\in aH=H$ and $\langle A^{-1}A\rangle\subseteq H<G$ contradicting with $2$, or $1\notin aH$ contradicting $1\in A^{-1}A$.
\item $3\Rightarrow 2$. Follows by definition.
\end{itemize} 

\end{proof}

\begin{proof}[Proof of \Cref{thm:gapped}]

Since $\mu$ is strongly scrambling, then by ~\Cref{prp:weakly_scrambling} we know that $\mu^{-1}*\mu$ is scrambling. Without loss of generality, assume the former. By~\Cref{thm:inverse_symm_gapped},
\begin{equation}
    |||\Lambda(\mu^{-1} * \mu) - \Lambda(\mu_H)|||_2 <1.
\end{equation}
That is, the third largest magnitude of the spectrum of $\Lambda(\mu^{-1}* \mu)$ is less than 1. Since $\Lambda(\mu^{-1}*\mu) = \Lambda(\mu)^\dagger \Lambda(\mu)$, we conclude the third largest singular value of $\Lambda(\mu)$ is less than 1, which is the desired conclusion.
\end{proof}

\bibliographystyle{unsrt}
\bibliography{ref}
\end{document}